\documentclass[5p,final]{elsarticle}
\usepackage{amsfonts}
\usepackage{amssymb}
\usepackage{amsmath}  
\usepackage{amsthm}
\usepackage{enumitem} 
\usepackage{geometry}
\usepackage{graphicx}
\usepackage{algorithm}
\usepackage{algcompatible}
\usepackage{algpseudocode}
\usepackage[colorlinks=true, allcolors=blue]{hyperref}
\usepackage{hyperref}
\usepackage[utf8]{inputenc} 
\usepackage[english]{babel} 
\usepackage{url}
\usepackage{wrapfig}
\usepackage{graphicx}
\usepackage{subfiles}
\usepackage{xfrac}
\usepackage{multirow}
\usepackage{cleveref}
\usepackage{balance}
\usepackage{svg}
\usepackage{booktabs}
\usepackage{tikz}
\usetikzlibrary{shapes.geometric, arrows, arrows.meta}
\usepackage{tikz-3dplot}
\usepackage[squaren]{SIunits}
\usepackage{caption}
\usepackage{multicol}
\usepackage{makecell}
\usepackage{rotating}
\usepackage{tabularx}
\usepackage[normalem]{ulem}
\usepackage{subcaption}
\usepackage{lipsum}
\usepackage{todonotes}
\usepackage{float}
\usepackage{mathtools}

\newtheorem{definition}{Definition}

\newcommand{\bc}{\begin{center}}
\newcommand{\ec}{\end{center}}
\newcommand{\be}{\begin{equation}}
\newcommand{\ee}{\end{equation}}
\newcommand{\bea}{\begin{eqnarray}}
\newcommand{\eea}{\end{eqnarray}}
\newcommand{\beq}{\begin{eqnarray*}}
\newcommand{\eeq}{\end{eqnarray*}}
\newcommand{\bv}{\left( \begin{array}{c} }
\newcommand{\ev}{\end{array} \right) }

\usepackage{natbib}

\newtheorem{proposition}{Proposition}
\newtheorem{remark}{Remark}

\begin{document}
\title{A Simple Hierarchical Causality Primer}

\author[uct-sta]{Tim Gebbie}
\ead{tim.gebbie@uct.ac.za}
\address[uct-sta]{Department of Statistical Science, University of Cape Town, Rondebosch 7701, South Africa}

\date{\today}

\begin{abstract}
We provide a brief primer for the idea behind formalising hierarchical causality in the context of complex systems. Here 
actors are not simply agents. Actors instantiate causation classes. Agents implement local dynamics in given levels of organisation in a given system. Hierarchical causality then describes how actor-level roles constrain, select, and organise agent-level behaviour across levels. The system then necessarily requires three additional structures. First, causation classes to abstract a given form of causal influence that an actor instantiates. Second, aggregation operators to move across the levels.  Third, discrete event-time maps are required because the system comprises events, and the relation between local event counts and any global clock must be specified. Our formulation here is purposefully simple and discrete. 
\end{abstract}

\maketitle
\tableofcontents

\section{Introduction}\label{sec:introduction}

Here the levels of organization are where things happen and where agents are interacting components. Actors are then carriers of top-down influence and are grouped into causation classes which are kinds of top-down causal operation equivalence classes. This is so that sets of lower-level realizations are treated as functionally the same. Here actors are a computational model used to characterize a specific type of top-down causation. Actors differ from agents the way top-down causation differs from bottom-up interacting components. 

In summary, an agent is a local, interacting unit in the system: trader, order, fund, node, borrower, or lattice element. An actor is a role-bearing or control-bearing entity that exemplifies a top-down causal mode. \footnote{In the original actor model, computation is organised as autonomous entities that encapsulate state and interact solely through asynchronous message passing and the creation of new actors \cite{Hewitt1973Actor,Hewitt1977Control}. Here, we use the term more loosely to represent a causal equivalence class -- this is purposefully referring to the prior equivalence class as a computational representation distinct from agents.} Think predictor, profiteer, investor, trader, regulator/ruler as specific causal-role archetypes \citep{wilcox2014hierarchical}, and not merely physical persons. Causation classes are then the abstract form of causal influence that the actor instantiates and is implemented as this or that agent in some level in the hierarchy. The causation classes are minimally described here as: algorithmic structuring, fixed-goal control, adaptive selection, adaptive feedback, or adaptive-goal selection \citep{auletta2008topdown,wilcox2014hierarchical}.

The reason why agents are not enough is that if you stay only at a given agent level, you can describe interactions, but you miss why the same sort of agent behaves differently across contexts. So the argument is basically that agents provide the local mechanics; actors represent organized causal roles; causation classes specify the kind of top-down influence; the system outcome depends on all of these interacting across levels. This implies the need for aggregation operators that move across levels, and transition kernels that describe within-level dynamics. The equivalence class structure allows lower-level realizations to be treated as functionally the same from the higher-level standpoint - these are what make higher-level causation possible without needing one privileged micro-realization.

The aim of the primer is to try to more formally make the distinction between top-down causation as constraints on the system as opposed to bottom-up aggregation driven emergence via mechanisms and algorithms. To then provide a simple toy-model with two levels: a lower level and an upper level for illustrative purposes \ref{app:minimal-discrete-example}. A more sophisticated agent-based trading example can be found when one introduces adaptive feed-back learning in trading simulations \cite{dicks2024}. 
More generally, one can expect to formulate equivalence classes of agent-based models using the language of discrete-time random walks, one for each layer. Under appropriate scaling limits and continuous-time random time changes, these may lead to coupled reaction--diffusion descriptions, with cross-level coupling entering through constraints or boundary conditions.

\section{Background}\label{sec:background}

Top-down causation is not new. The useful starting point is the observation that complex systems have levels, and that higher-level organisation can matter without requiring a unique lower-level realisation. In this sense the higher level does not replace the lower level. It constrains it, selects among its admissible realisations, or changes which lower-level dynamics are relevant. This is the main idea we retain from the top-down causation literature \citep{ellis2012topdown,auletta2008topdown}. This is also what makes it quite distinct from much of the prevailing complexity literature which is by-and-large bottom-up and reductionist.

The important point is not the phrase ``top-down''. It is the existence of causal equivalence classes. A higher-level state, function, or role may be realised by many lower-level configurations. These lower-level configurations can be different as microstates and still be the same for the higher-level causal question. This is why aggregation alone is not enough. Aggregation produces a macro-description. Hierarchical causality asks whether that macro-description also carries a causal role for the lower-level transition structure.

There is also related work on causal abstraction and causal emergence. Structural causal-model approaches ask when two causal descriptions at different resolutions are consistent under interventions \citep{rubenstein2017causal,beckers2019abstracting}. Causal-emergence approaches ask when a macro-description can be causally stronger, or less noisy, than a micro-description \citep{hoel2013quantifying}. These literatures are related -- but what we are doing is a bit narrower. We are not trying to replace structural causality, causal abstraction, or causal-emergence measures. We want to use a discrete kernel schema to separate three things that are often collapsed: aggregation across levels, transition within levels, and actor-class constraints on admissible transitions. The schema is also adjacent to controlled Markov processes and constrained Markov decision processes \citep{puterman1994mdp,altman1999cmdp}. The difference is that the actor instance here is not primarily an optimiser over rewards or costs. It is a role-bearing cross-level constraint that restricts admissible lower-level kernels through an interface. In this sense the framework borrows the language of admissible controlled dynamics, but not the usual optimisation problem of Markov decision theory.

The approach here is really to pry these objects apart, but is otherwise modest. \footnote{I will very loosely draw on ideas from category theory \cite{spivak2014category}.} A hierarchical causal system is not just a multiscale model. It is a multiscale model with role-bearing actors and causation classes that act on the lower-level dynamics through constraints, selection, feedback, or goal-directed kernel modification. The primer tries to make this statement precise enough to be useful but not to overly mathematise the idea itself. 

\section{Hierarchical Causality}\label{sec:core-definition}

The point here is to keep hierarchy, dynamics, constraint, and event time separate. If these are collapsed into one mechanism then hierarchical causality becomes either ordinary coarse-graining or ordinary agent dynamics. Here for brevity we are working with finite or countable state spaces, so that transition kernels can be written as conditional probabilities and evaluated state by state. A more general measurable-space version would replace these probabilities by Markov kernels on measurable spaces \citep{kallenberg2002foundations}.

\begin{definition}[Hierarchical causal system]
\label{def:hierarchical}
A discrete hierarchical causal system is a tuple
\[
\Psi=\{H,D,C,U\},
\]
with at least four elements. Here $H$ denotes the structure required to specify the hierarchy, $D$ denotes the structure required to specify the dynamics, $C$ denotes the top-down constraint structure, and \(U\) denotes the discrete event-time maps used to relate local event counts to each other and, when needed, to a global clock.

\begin{enumerate}[label=(\roman*)]
    \item The hierarchy is the collection $H$ that contains the levels, configuration spaces, agent index sets, accessible state spaces, and aggregation maps. Thus
    \[
    H=\left\{\{X_{\ell}\}_{\ell=0}^{L},\{I_{\ell}\}_{\ell=0}^{L},\{S_{\ell}\}_{\ell=0}^{L},\{\Pi_{\ell}\}_{\ell=0}^{L-1}\right\}.
    \]
    The index $\ell\in\{0,\ldots,L\}$ labels levels of organisation, with lower $\ell$ denoting the more microscopic description. Here $I_{\ell}$ denotes the set of level-specific sites, agents, or components, and $S_{\ell}$ denotes the accessible state space at level $\ell$. For the present discrete formulation $S_{\ell}$ may be taken finite or countable; more general measurable state spaces require Markov kernels. A level state at operational event count $n_{\ell}$ is
    \[
    X_{\ell,n_{\ell}}
    =
    \left(x^{\ell}_{i,n_{\ell}}\right)_{i\in I_{\ell}}
    \in X_{\ell},
    \qquad
    x^{\ell}_{i,n_{\ell}}\in S_{\ell}.
    \]
Here $X_{\ell}$ is the level-$\ell$ configuration space. When every assignment of state values to level-$\ell$ sites or agents is admissible, $X_{\ell}=S_{\ell}^{I_{\ell}}$; otherwise $X_{\ell}\subseteq S_{\ell}^{I_{\ell}}$ is the admissible configuration space. It is important to notice that each level $\ell$ has its own sites, agents, state spaces, and event counts, denoted here by $I_{\ell}$, $S_{\ell}$, and $n_{\ell}$. Aggregation is given by
    \[
\Pi_{\ell}:X_{\ell} \rightarrow X_{\ell+1}.
    \]
    This map need not be an average. It may be a coarse-graining, projection, threshold, classification, or other map that makes a higher-level state from lower-level realisations.

    \item The dynamics are given by \(D\), which contains the admissible within-level transition structures. For each level there is a class
\(\mathcal{K}_{\ell}\) of admissible discrete kernels, and a realised kernel
\[
K_{\ell,n_{\ell}}\in\mathcal{K}_{\ell},
\]
with transition probabilities written as
\(K_{\ell,n_{\ell}}(x'_{\ell}\mid x_{\ell})\). This is the bottom-up part of
the system. Agents interact locally, and higher-level regularities may emerge through aggregation of these dynamics. The kernel is indexed by the
level-specific event count \(n_\ell\). The coordination of event counts across levels is specified separately by \(U\).

    \item The constraints are encapsulated in \(C\), which contains level-indexed actor sets \(\{\mathcal{A}_r\}_{r=0}^{L}\), causation classes \(\mathcal{C}\), target subsystems, and interfaces. Here \(\mathcal{A}_r\) denotes the set of actors represented at level \(r\). An actor \(a\in\mathcal{A}_r\) instantiates a causation class through
\[
\gamma:\bigcup_{r=0}^{L}\mathcal{A}_{r}\rightarrow\mathcal{C}.
\]
The operative cross-level object is an actor instance
\[
\alpha=(a,r,c,\ell,B_{\ell},\eta),
\qquad c=\gamma(a),\quad r>\ell,
\]
where \(B_{\ell}\subseteq I_{\ell}\) is the lower-level target subsystem, \(\eta\) is the interface, channel, rule, signal, protocol, incentive, or admissibility test through which the higher-level role is implemented, and \(c\) is the causation class. When a target subsystem \(B_\ell\) is being constrained, write \(\mathcal{K}_{\ell,B_\ell}\) for the relevant local,
restricted, marginal, or conditional kernel class on \(B_\ell\). The
corresponding target-subsystem section of a realised kernel
\(K_{\ell,n_\ell}\) is denoted by \(K^{B_\ell}_{\ell,n_\ell}\).

The actor instance induces a set-valued constraint correspondence
\[
D^{\alpha}_{r\rightarrow\ell}:X_{r}\rightrightarrows
\mathcal{K}_{\ell,B_{\ell}}.
\]
For a higher-level state \(X_{r,n_{r}}\), admissibility of the realised
lower-level dynamics on the target subsystem means
\[
K^{B_\ell}_{\ell,n_{\ell}}
\in
D^{\alpha}_{r\rightarrow\ell}(X_{r,n_{r}}).
\]
Top-down causation is located here. It is not aggregation and it is not the
lower-level dynamics themselves. It is the constraint, selection, weighting, or
parameterisation of the target-subsystem transition structure by a higher-level
actor role through an interface.

    \item The time structure \(U\) contains discrete event-time maps. We use the
term discrete subordinator as a non-decreasing map from a global event index to
a level-specific event count, not necessarily a Lévy subordinator. Each level
has its own operational event count \(n_{\ell}\), and a global event index
\(m\in\mathbb{N}\) coordinates the levels through non-decreasing maps
\[
n_{\ell}=U_{\ell}(m),\qquad U_{\ell}(0)=0.
\]
Different levels may therefore update on different event scales. If calendar
time is later required, it may be introduced by an additional non-decreasing
time change from event counts to calendar time. This is not part of the basic
discrete definition. How events happen and are counted is typically determined
bottom-up in tandem with the constraints.
    \item Aggregation equivalence is defined by:
    \[
x_{\ell}\mathrel{\underset{\Pi_{\ell}}{\sim}} y_{\ell}\quad\Longleftrightarrow\quad\Pi_{\ell}(x_{\ell})=\Pi_{\ell}(y_{\ell}).
    \]
    Aggregation equivalence may identify states that are not causally equivalent.
\item Causal equivalence is actor-specific and is stronger for the causal question being asked. For a fixed actor instance \(\alpha\), higher-level state \(x_r\), and lower-level state \(x_\ell\), define the admissible outgoing target-subsystem kernel sections
\[
\mathcal{D}^{\alpha}_{x_r}(x_\ell)
:=
\{K^{B_\ell}(\cdot\mid x_\ell):K^{B_\ell}\in
D^{\alpha}_{r\rightarrow\ell}(x_r)\}.
\]
Here \(K^{B_\ell}(\cdot\mid x_\ell)\) is understood as the relevant local, restricted, marginal, or conditional outgoing kernel section on \(B_\ell\).
Then two lower-level states are causally equivalent for \(\alpha\) at \(x_r\) when
\[
x_\ell\sim^{\alpha}_{x_r} y_\ell
\quad\Longleftrightarrow\quad
\mathcal{D}^{\alpha}_{x_r}(x_\ell)
=
\mathcal{D}^{\alpha}_{x_r}(y_\ell).
\]
Thus aggregation equivalence may identify states that are not causally equivalent for the actor instance being considered.
\end{enumerate}

The system is hierarchically causal when there exists at least one actor instance \(\alpha=(a,r,c,\ell,B_{\ell},\eta)\), with \(r>\ell\), such that the realised lower-level kernel, through its target-subsystem section, satisfies
\[
K^{B_\ell}_{\ell,n_{\ell}}
\in
D^{\alpha}_{r\rightarrow\ell}(X_{r,n_{r}}),
\]
and this restriction is not determined by aggregation alone.
\end{definition}

\begin{proposition}[Aggregation equivalence need not imply causal equivalence] \label{prop:equivalence}
Let \(x_\ell,y_\ell\in X_\ell\) satisfy
\[
\Pi_\ell(x_\ell)=\Pi_\ell(y_\ell).
\]
If there exists an actor instance \(\alpha\) and a higher-level state \(x_r\) such that
\[
\mathcal{D}^{\alpha}_{x_r}(x_\ell)
\neq
\mathcal{D}^{\alpha}_{x_r}(y_\ell),
\]
then \(x_\ell\) and \(y_\ell\) are aggregation equivalent but not causally equivalent for \(\alpha\) at \(x_r\).
\end{proposition}

\begin{proof}
The equality \(\Pi_\ell(x_\ell)=\Pi_\ell(y_\ell)\) gives aggregation equivalence. The displayed inequality says that the admissible outgoing kernel sections differ under the actor instance \(\alpha\). Hence the two states are not causally equivalent for \(\alpha\) at \(x_r\).
\end{proof}

\begin{remark}
The actor is not a free-floating cross-level force. It is represented at a higher level, often as an agent, institution, rule system, coalition, or role-bearing structure at that level. Its top-down role is expressed only through the interface $\eta$ and the induced constraint on admissible lower-level kernels.
\end{remark}

\begin{remark}
A kernel transformation is a special implementation of the constraint structure, not the definition of top-down causation itself. In this special case one may write a transformation from an unconstrained kernel to a constrained kernel. But the more general object is in-fact the admissible set of kernels selected by the higher-level actor-class constraints.
\end{remark}

\begin{remark}
This is close in spirit to the viability-kernel idea from which some of this thinking is derived. Viability theory studies evolutions that remain inside a constraint set, and the viability kernel is the set of initial states from which at least one viable evolution remains possible \citep{aubin2011viability,saintpierre1994viability}. Here the object is not yet a viability kernel over states. It is a constraint correspondence over admissible transition kernels. The analogy is useful because it keeps the top-down component constraint-based rather than force-based.
\end{remark}

\begin{remark}
The definition is intentionally discrete. Nothing here requires continuous time. A continuous-time version would require limits of the kernels and subordinators, and is left to another day.
\end{remark}

\section{Structural Components}\label{sec:components}

Definition \ref{def:hierarchical} provides the basic idea of what one means when one uses the terminology hierarchical causality in the sense of accommodating top-down causation. Here I briefly expand on some of the other terminology used for clarity.

\begin{definition}[Agent]
An agent is a local state-bearing component in one level of the hierarchy. At level $\ell$ and event count $n_{\ell}$ the state of the $i$-th site, agent, or component is written as $x^{\ell}_{i,n_{\ell}}\in S_{\ell}$, where $i\in I_{\ell}$. Agents carry local dynamics through the realised kernel in $D$. They are not, by themselves, top-down causes.
\end{definition}

\begin{definition}[Actor and actor instance]
An actor is a role-bearing entity represented at some level $r$. It is written $a\in\mathcal{A}_{r}$. It may itself be an agent, institution, coalition, rule system, or other state-bearing structure at level $r$. The actor becomes a top-down causal object only as an actor instance: $\alpha=(a,r,c,\ell,B_{\ell},\eta)$, 
where $c=\gamma(a)$ is the causation class, $\ell<r$ is the target level, $B_{\ell}\subseteq I_{\ell}$ is the target subsystem, and $\eta$ is the interface through which the constraint is implemented.
\end{definition}

\begin{definition}[Causation class]
A causation class is an abstract mode of top-down constraint. Here the main classes are: algorithmic structuring, fixed-goal control, adaptive selection, adaptive feedback, and adaptive-goal selection. Formally, a class is used through the constraint block $C$: if $c=\gamma(a)$, then the actor instance $\alpha$ determines which lower-level kernels remain admissible under the higher-level state.
\end{definition}

\begin{definition}[Transition kernel]
A transition kernel at level $\ell$ is a discrete one-step law
$K_{\ell,n_{\ell}}\in\mathcal{K}_{\ell}$, with transition probabilities written as $K_{\ell,n_{\ell}}(x'_{\ell}\mid x_{\ell})$. It belongs to the dynamics structure $D$. The kernel describes admissible within-level transitions. When top-down constraint is active, the relevant target-subsystem section
\(K^{B_\ell}_{\ell,n_\ell}\) must lie in
\(D^{\alpha}_{r\rightarrow\ell}(X_{r,n_r})\).
\end{definition}

\begin{definition}[Aggregation operator]
An aggregation operator is a map: $\Pi_{\ell}:X_{\ell}\rightarrow X_{\ell+1}$.
It makes a higher-level description from lower-level realisations. It need not be an average. It may be a coarse-graining, projection, threshold, classification, or other description map. Aggregation can generate macro variables and macro regularities, but aggregation alone is not top-down causation.
\end{definition}

\begin{definition}[Discrete subordination]
A discrete subordinator is a non-decreasing map $U_{\ell}:\mathbb{N}\rightarrow\mathbb{N}$ with $U_{\ell}(0)=0$. It links the global event index $m$ to the operational event count $n_{\ell}=U_{\ell}(m)$ at level $\ell$. This lets different levels update at different event scales without introducing continuous time.
\end{definition}

\begin{definition}[Aggregation and causal equivalence]
Aggregation equivalence is the equivalence relation induced by \(\Pi_\ell\). Actor-specific causal equivalence is the corresponding relation induced by the constrained admissible transition structure in Definition \ref{def:hierarchical}. Hence aggregation equivalence can hold while causal equivalence fails.
\end{definition}

\section{Equivalence}\label{sec:equivalence}

The formal schema gives two routes through the hierarchy, this is shown visually in Figure \ref{fig:equivalence-diagram}. The first route is descriptive: evolve at the lower level and then aggregate up. The second route is constrained: use a higher-level actor instance to restrict the lower-level kernel, evolve under that restriction, and only then aggregate. If these routes carry the same information for the question being asked, aggregation equivalence is enough. If not, the difference is precisely where hierarchical causality matters \citep{wilcox2014hierarchical}.

\begin{figure*}[ht!]
\centering
\begin{tikzpicture}[
  x=1cm,y=1cm,>=Latex,
  every node/.style={font=\small},
  block/.style={align=center,minimum width=3.55cm,minimum height=1.08cm},
  lab/.style={font=\small,fill=white,inner sep=1.5pt}
]
    \node[block] (xr)     at (0,3.2) {$X_{r,n_r}$\\ higher-level state};
    \node[block] (xrdesc) at (8.4,3.2) {$\Pi_{\ell}(X_{\ell,n_{\ell}+1})$\\ aggregate description};
    \node[block] (xl)     at (0,0) {$X_{\ell,n_{\ell}}$\\ lower-level state};
    \node[block] (xlp)    at (8.4,0) {$X_{\ell,n_{\ell}+1}$\\ lower-level state};
    \draw[->] (xl.east) -- node[lab,below=3pt] {$K_{\ell,n_{\ell}}$} (xlp.west);
    \draw[->] (xl.north) -- node[lab,left=3pt] {$\Pi_{\ell}$} (xr.south);
    \draw[->] (xlp.north) -- node[lab,right=3pt] {$\Pi_{\ell}$} (xrdesc.south);
    \draw[->] (xr.east) -- node[lab,above=3pt] {descriptive aggregate route} (xrdesc.west);
\draw[->,dashed,bend left=12]
  (xr.south east)
  to node[lab,right,align=center]
  {$D^{\alpha}_{r\rightarrow\ell}$\\ constraint on $K^{B_\ell}_{\ell,n_{\ell}}$}
  (xlp.north west);
    \node[align=center,font=\small] at (4.2,-1.15)
      {Non-commutation marks the gap between aggregation equivalence and causal equivalence.};
\end{tikzpicture}
\caption{Aggregation and constraint in a hierarchical causal system. The higher-level structure is shown above the lower-level dynamics. The horizontal lower arrow is lower-level dynamics, labelled by the realised lower-level kernel \(K_{\ell,n_{\ell}}\). The vertical arrows are aggregation. The dashed arrow is the top-down constraint induced by an actor instance \(\alpha=(a,r,c,\ell,B_{\ell},\eta)\), acting on the target-subsystem kernel section \(K^{B_\ell}_{\ell,n_{\ell}}\). The diagram is not assumed to commute. When it fails to commute, aggregation equivalence is weaker than causal equivalence.}
\label{fig:equivalence-diagram}
\end{figure*}

The dashed arrow is not a new dynamics. It is a restriction on admissible lower-level kernels. Thus the constrained kernel must satisfy
\[
K^{B_\ell}_{\ell,n_{\ell}}\in
\mathcal{K}_{\ell,B_{\ell}}\cap
D^{\alpha}_{r\rightarrow\ell}(X_{r,n_r}).
\]
If this condition changes the lower-level transitions available to the target subsystem, and if the change is not determined by aggregation alone, then the system carries top-down causal structure.

The distinction can also be stated without a diagram. Let $x_{\ell}$ be a lower-level state, and let $K_{\ell}$ be an unconstrained lower-level kernel. Aggregation gives the higher-level description $\Pi_{\ell}(x_{\ell})$. Bottom-up emergence concerns the behaviour of this aggregate under lower-level dynamics. Top-down causation enters only when a higher-level actor instance $\alpha$ restricts the relevant target-subsystem section of the lower-level kernel to $D^{\alpha}_{r\rightarrow\ell}(X_{r})$.

This is why aggregation equivalence is weaker than causal equivalence. If $x_{\ell}$ and $y_{\ell}$ have the same aggregate image, they are the same for a purely descriptive macro-variable. They need not be the same for a top-down causal question, because they may admit different constrained transition sets under the same actor instance. This is the non-commuting part of the hierarchy. Aggregating first and then evolving need not give the same causal information as constraining the lower-level dynamics and then aggregating. In many ways this is the entire point of the exercise.

\section{Formal schema}\label{sec:formal-schema}

In summary, the formal schema is
$\Psi=(H,D,C,U)$.

\subsection{The Hierarchy}
The hierarchy structure $H$ fixes the level-indexed descriptions. For each level $\ell$ there is a configuration space $X_{\ell}$, an index set $I_{\ell}$ of sites, agents, or components, an accessible state space $S_{\ell}$, and a state
\[
X_{\ell,n_{\ell}}
=
\left(x^{\ell}_{i,n_{\ell}}\right)_{i\in I_{\ell}}
\in X_{\ell},
\qquad
x^{\ell}_{i,n_{\ell}}\in S_{\ell}.
\]
Here $X_{\ell}$ is the level-$\ell$ configuration space. When every assignment of state values to level-$\ell$ sites or agents is admissible, $X_{\ell}=S_{\ell}^{I_{\ell}}$; otherwise $X_{\ell}\subseteq S_{\ell}^{I_{\ell}}$ is the admissible configuration space. The aggregation operator $\Pi_{\ell}$ maps from level $\ell$ to level $\ell+1$. It defines the descriptive macro-state. It does not define the top-down constraint. Aggregation equivalence $x_{\ell}{\sim}y_{\ell}$ is the equivalence relation induced by $H$ given $\Pi_{\ell}$.

\subsection{The Dynamics}
The dynamics $D$ fixes the within-level transition classes. For each level there is a class of admissible kernels $\mathcal{K}_{\ell}$. A realised kernel is $K_{\ell,n_{\ell}}\in\mathcal{K}_{\ell}$, with transition probabilities $K_{\ell,n_{\ell}}(x'_{\ell}\mid x_{\ell})$.
The lower-level dynamics can generate higher-level regularities through aggregation. This is the bottom-up direction. Here we have kept this discrete and there is no continuous-time limit in the formal schema. When a target subsystem $B_{\ell}\subseteq I_{\ell}$ is being constrained, we write $\mathcal{K}_{\ell,B_{\ell}}$ for the relevant kernel class on that subsystem. This may be a marginal, restricted, local, or conditional kernel class, depending on the example.

\subsection{The Constraints}
The constraints $C$ are where top-down causation enters. An actor instance is $\alpha=(a,r,c,\ell,B_{\ell},\eta)$, with $c=\gamma(a)$, and $r>\ell$ induces a correspondence
\[
D^{\alpha}_{r\rightarrow\ell}: X_{r}\rightrightarrows \mathcal{K}_{\ell,B_{\ell}}.
\]
Thus the higher-level state does not determine the lower-level state. It restricts the lower-level transition possibilities. At event counts $n_r$ and $n_\ell$, admissibility on the target subsystem means
\[
K^{B_\ell}_{\ell,n_{\ell}}
\in
D^{\alpha}_{r\rightarrow\ell}(X_{r,n_{r}}).
\]
This is the formal constraint condition.

The interface $\eta$ is deliberately broad. It may be a rule, signal, protocol, incentive, admissibility test, boundary condition, selection criterion, or feedback channel. It is included so that the actor does not become a mysterious cross-level force. The actor is represented at level $r$; the interface implements the constraint on level $\ell$.

\subsection{Time Emergence}

The time structure coordinates event counts. A global event index $m\in\mathbb{N}$ is mapped to level-specific event counts by
$n_{\ell}=U_{\ell}(m)$, with $U_{\ell}(0)=0$,
where each $U_{\ell}$ is non-decreasing. A lot more could be said here but it lets one level update more slowly or more irregularly than another. It also lets a higher-level constraint be evaluated at $n_r=U_r(m)$ while a lower-level kernel acts at $n_\ell=U_\ell(m)$.

\subsection{Causal consistency}

A minimal consistency condition is that the realised lower-level kernel, understood locally on the target subsystem where needed, is both dynamically admissible and constraint admissible. For an actor instance $\alpha$ this means
\[
K^{B_\ell}_{\ell,n_{\ell}}\in
\mathcal{K}_{\ell,B_{\ell}}\cap
D^{\alpha}_{r\rightarrow\ell}(X_{r,n_{r}}).
\]
If this intersection is empty, the actor-state and the lower-level dynamics are incompatible at that event. If it is non-empty, top-down causation is represented by choosing, selecting, or realising a kernel inside the constrained admissible set.

The hierarchical causality argument is therefore not that higher-level states cause lower-level states directly. The claim is weaker and cleaner: higher-level actor instances restrict the admissible lower-level transition structure, and this restriction is not determined by aggregation alone.

\section{Discussion}\label{sec:discussion}

The point of the construction is modest. It separates three operations that are often mixed together. Aggregation gives a higher-level description. Dynamics gives the lower-level production of future states. Constraint gives the top-down restriction of admissible lower-level transitions. A hierarchical causal system needs all three, and it also needs event-time bookkeeping so that levels need not update on the same clock. The actor-instance notation is useful. A higher-level actor is not assumed to act as a mysterious cross-level force. The actor is represented at a level, and its causal role is expressed through an interface that restricts a lower-level target subsystem. The main idea here is to ensure that causal equivalence is not confused with aggregation equivalence {\it e.g.} two lower-level states may have the same aggregate description and still differ causally because a higher-level actor instance may make different lower-level kernels admissible. Top-down causation can start to look like interventions, so we contrast this formulation with some of the existing formulations.

Structural causal models in the Pearl tradition give the sharpest intervention language. Their core idea is that causality is represented by structural assignments, directed graphs, interventions, and counterfactuals; the central question is what changes under an operation such as $do(X=x)$ \citep{pearl2009causality}. This is stronger than the present framework for identification, counterfactual analysis, and empirical causal inference. For the present purpose, however, structural causal models do not by themselves separate aggregation, event-time coordination, actor interfaces, and top-down restrictions on admissible lower-level kernels. In short, Pearl gives intervention semantics; this framework is basically a form of hierarchically constrained bookkeeping -- it is algorithmic by nature. The key issue still remains how best to understand the aggregation which is related to causal emergence.

The aggregation maps used here are also related to coarse-graining in statistical physics and multiscale modelling \citep{goldenfeld1992lectures}. The important distinction is that coarse-graining supplies a macro-description, while hierarchical causality asks whether an actor instance restricts the admissible lower-level transition structure. Aggregation is therefore necessary for the hierarchy, but it is not itself the top-down causal operation.

Hoel's causal-emergence programme asks when a macro-description can be more causally informative than a micro-description. Its core idea is that coarse-graining may reduce degeneracy or noise, so that causal power can be greater at the macro level than at the micro level \citep{hoel2013quantifying}. This is close to our concern with aggregation, but it asks a different question. Hoel gives a measure of macro causal strength. The present framework instead asks whether a higher-level actor instance restricts the admissible lower-level transition structure. Hoel is stronger when the state spaces, interventions, and transition probabilities are explicit enough for measurement. The present framework is more directly suited to systems where the key causal object is a rule, institution, protocol, or interface that constrains local dynamics but is not naturally a single macro variable. At the heart of this is really the role of abstraction.

Causal abstraction approaches ask when causal models at different resolutions are consistent. The core idea is that a macro causal model should preserve the relevant intervention structure of a micro causal model under an abstraction map \citep{rubenstein2017causal,beckers2019abstracting}. This is probably the closest formal neighbour to the aggregation part of the present framework. It is stronger in its treatment of consistency between causal models. The difference is that our object is not only an abstraction map from micro to macro variables. It also contains actor instances, target subsystems, interfaces, and constraint correspondences. Causal abstraction tells us when levels agree under intervention; this framework says how higher-level roles can restrict admissible lower-level kernels.

Controlled Markov processes and constrained Markov decision processes form another nearby reference class \citep{puterman1994mdp,altman1999cmdp}. They also work with transition kernels, admissible controls, and constraints on dynamic evolution. The difference is again one of purpose. In a constrained MDP the central object is usually a controller or policy that optimises a reward or cost criterion subject to constraints. Here the actor instance is not introduced as an optimiser. It is a role-bearing structure represented at a higher level, and its causal role is to restrict which lower-level kernels are admissible through the interface \(\eta\). Thus the present framework is closer to constrained-dynamics bookkeeping than to stochastic optimal control. It may later be specialised into a controlled Markov model, but that is not required by the definition.

It is for this reason that we have tried to follow the Ellis-style top-down causation. Its core idea is that higher-level organisation can matter causally through constraints, information control, multiple realisability, and equivalence classes \citep{ellis2012topdown,auletta2008topdown}. The strength of this approach is realism: it speaks naturally about rules, biological functions, institutions, adaptive control, and organised contexts. Its weakness is that the formal objects are often less explicit. We think this realism is important. The present framework should be read as a small descriptive discrete formalisation of this intuition. 
The actor instance \(\alpha=(a,r,c,\ell,B_{\ell},\eta)\), together with \(D^{\alpha}_{r\rightarrow\ell}\), is a way of making the Ellis-style constraint story operational. The higher level is not treated as a second physical force or as a magical intervention. The action is in the restriction of lower-level dynamics.

It is here that I think viability theory is helpful because it gives a mature language for constrained dynamics. Its core idea is that a system evolves subject to admissibility constraints, and the viability kernel identifies states from which at least one viable evolution remains possible \citep{aubin2011viability,saintpierre1994viability}. This is useful because it keeps the language constraint-based rather than force-based. But the present framework is not computing a viability kernel. Its constraint object is a correspondence over admissible lower-level kernels, not primarily a set of viable initial states. Viability theory is stronger when the constraint set and dynamics are explicit enough for computation. Our framework is more schematic, but it includes actor roles, interfaces, and hierarchical aggregation explicitly.

The present framework is deliberately narrower than these alternatives. The approach is architectural rather than ingredient-level. The ingredients themselves are familiar: aggregation, kernels, constraints, actors, and event-time maps. The idea is to keep these objects separate and to locate top-down causation specifically in actor-instance restrictions on admissible lower-level kernels.

It is in this sense that a hierarchical causal system is defined as \(\Psi=(H,D,C,U)\) (Definition \ref{def:hierarchical}); where \(H\) carries aggregation, \(D\) carries lower-level dynamics, \(C\) carries actor-instance constraints, and \(U\) carries event-time coordination. Its strength is that it keeps these pieces separate. This can be considered weak or incomplete because we have deliberately avoided specifying how to identify causal effects from data, how to quantify macro causal strength, how to prove abstraction consistency, or compute viability kernels. This is by design because its realism lies in the middle: it is not a full inference theory, but a modelling grammar for systems in which higher-level rules, institutions, protocols, or adaptive roles restrict what lower-level dynamics are admissible.

\section{Conclusion}\label{sec:conclusion}

The basic idea is simple. Bottom-up emergence is produced by lower-level dynamics and aggregation. Top-down causation is represented by higher-level actor instances constraining the admissible lower-level transition structure. The restriction must not be determined by aggregation alone; and this is what prevents the construction from collapsing into ordinary coarse-graining. 

To show this we use Definition \ref{def:hierarchical} to define a discrete hierarchical causal system as $\Psi=(H,D,C,U)$. The hierarchy structure $H$ carries levels and aggregation. The dynamics structure $D$ carries within-level kernels. The constraint structure $C$ carries actor instances and their restrictions on admissible lower-level kernels. The time structure $U$ carries discrete event-time subordination. This is then made clear in Proposition \ref{prop:equivalence}, which although somewhat trite and tautological makes a very serious 
point why in systems with top-down causation aggregation equivalence need not imply causal equivalence. A short example is then provided in Appendix \ref{app:minimal-discrete-example}.

\appendix
\section{Small discrete example}\label{app:minimal-discrete-example}
This example is deliberately simple. It is not a domain model. The idea is only to show the four elements $H,D,C,U$ working together.

\subsection{The Hierarchy}
Let level $0$ contain four binary agents,
\[
X_{0,n}=(x^0_{1,n},x^0_{2,n},x^0_{3,n},x^0_{4,n})\in\{0,1\}^{4}.
\]
Let level $1$ contain two block states,
\[
X_{1,k}=(x^1_{1,k},x^1_{2,k})\in\{0,1,2\}^{2},
\]
where the aggregation map is the pair of block sums
\[
\Pi_{0}(X_{0,n})=(x^0_{1,n}+x^0_{2,n},x^0_{3,n}+x^0_{4,n}).
\]
Thus $x^1_{1,k}$ is the number of active agents in the first block and $x^1_{2,k}$ is the number of active agents in the second block. When the upper-level state is obtained by aggregation at a common global event index $m$, one reads $n=U_0(m)$ and $k=U_1(m)$.

\subsection{The Dynamics}
At each lower-level event one site is selected. Without top-down constraint, the selected bit flips with probability $p$. This defines a simple finite Markov kernel $K_{0,n}$ on $\{0,1\}^{4}$. The point is that the lower level has its own somewhat trivial admissible dynamics.

\subsection{The Constraints}

Let level \(1\) contain an actor \(a\) with class \(c=\gamma(a)\) and target subsystem \(B_{0}=\{1,2\}\). The actor instance is
\[
\alpha=(a,1,c,0,B_{0},\eta).
\]
The interface \(\eta\) is a block rule with site-level content: when the first block has exactly one active site, only the first site in that block may change. Thus, when \(x^0_{1,n}+x^0_{2,n}=1\), admissible kernels assign probability zero to transitions that flip the second site. The probability mass of forbidden flips is assigned to the self-transition, or equivalently the constrained kernel is understood as an admissible kernel with zero probability on the forbidden transitions.

This makes the distinction between aggregation equivalence and causal equivalence explicit. The two lower-level states
\[
(1,0,0,0)
\quad\text{and}\quad
(0,1,0,0)
\]
have the same aggregate image
\[
\Pi_{0}(X_{0,n})=(1,0),
\]
but the actor-induced admissible transitions differ because the interface refers to the lower-level site structure inside the block. Equivalently,
\[
K^{B_0}_{0,n}\in D^{\alpha}_{1\rightarrow 0}(X_{1,k})
\]
means that the target-subsystem section of the lower-level kernel is compatible not only with the aggregate block count but also with the actor-instance rule acting through the interface \(\eta\).

\subsection{Time Sampling}
Let the lower level update at every global event,
\[
U_{0}(m)=m,
\]
and let the higher level update every two lower-level events,
\[
U_{1}(m)=\lfloor m/2\rfloor.
\]
The higher-level state therefore changes more slowly. The constraint used by the lower-level kernel at event $m$ is evaluated using $X_{1,U_{1}(m)}$.

\subsection{What the example shows}
 
The example shows the point of the construction. Two lower-level states can have the same aggregate block count and yet differ in the admissible next-step transitions once an actor instance is active. Aggregation gives the description. The unconstrained bit-flip kernel gives bottom-up dynamics. The actor instance restricts the admissible lower-level kernel through an interface that can see structure hidden by aggregation. The discrete event-time maps specify which event count is used at each level.


\section{Simulation-domain example}
\label{app:simulation-requirements}

This appendix is a requirements specification rather than an implementation. The purpose is to give a bridge from the small formal primer to a more serious simulation study. The intended domain is an agent-based financial market with a limit order book, ordinary trading agents, adaptive learning agents, and higher-level actors such as index rebalancers, regulators, circuit-breaker rules, and institutional execution schedules. The point is not to claim that this is the only useful domain. It is just a compact domain in which aggregation, local dynamics, actor constraints, and event-time maps can all be made explicit.

The important modelling choice is that the experiment is constraint-driven. A higher-level actor is not represented as an additional force acting on lower-level agents. It is represented as a role-bearing object whose interface changes the set of admissible lower-level transition kernels. Thus the simulation should test whether an actor instance
\[
\alpha=(a,r,c,\ell,B_\ell,\eta)
\]
induces a restriction
\[
K^{B_\ell}_{\ell,n_\ell}
\in
D^{\alpha}_{r\rightarrow \ell}(X_{r,n_r})
\]
that has effects not recoverable from aggregation-equivalent lower-level perturbations alone. This keeps the simulation aligned with the constraint view of top-down causation in Ellis-style and Auletta--Ellis--Jaeger-style accounts, and with the broader constraint-based account of causation in complex systems \citep{ellis2012topdown,auletta2008topdown,juarrero2023context}.

\subsection{Mapping to the schema}
\label{app:model-specification-requirements}

The simulation must instantiate the four pieces of the hierarchical schema
\[
\Psi=(H,D,C,U).
\]
For the market example this should be done as a concrete model specification, not just as a list of concepts. The lower level contains the limit order book, outstanding orders, agent inventories, private or latent signals, agent policy states, and event-level order flow. The upper level contains market observables and actor states. The aggregation map, written here as $G$ or $\Pi_0$, maps lower-level configurations into macro variables such as midprice, spread, depth, realised volatility, signed volume, and order-flow imbalance. The constraint block contains actor instances. The time block records order-arrival time, learning-update time, actor-event time, and any reporting clock.

A minimally useful implementation should therefore specify the components in Table~\ref{tab:simulation-components}. This table is part of the requirements. It prevents the market simulator from becoming only an informal agent-based model with a few macro variables added after the fact. It is also close in spirit to the ODD discipline for agent-based modelling, where model description, design concepts, implementation detail and replication conditions must be made explicit enough for others to inspect the model rather than only its outputs \citep{grimm2020odd}.

\begin{table*}[ht!]
\centering
\footnotesize
\caption{Required model components for a hierarchical-causality market simulation. The table gives the concrete simulation object and its role in the schema $\Psi=(H,D,C,U)$.}
\label{tab:simulation-components}
\begin{tabularx}{\textwidth}{p{0.16\textwidth}p{0.24\textwidth}p{0.24\textwidth}X}
\toprule
Schema block & Simulation object & Minimum required content & Purpose in the experiment \\
\midrule
$H$ & Level descriptions and aggregation map $G$ & Micro state $M_t$ containing orders, inventories, agent states, signals, learning states; macro state $A_t=G(M_t)$ containing price, liquidity, volatility and flow variables & Fixes what is meant by aggregation, and what macro description is being tested \\
$D$ & Within-level market dynamics & LOB matching engine, order submission, cancellation, execution, inventory updates, signal updates, and learning updates where applicable & Defines the lower-level transition kernels before actor constraints are applied \\
$C$ & Actor instances and interfaces & Rebalancer, regulator, circuit-breaker, institutional-flow or other role-bearing actor; target subsystem $B_\ell$; interface $\eta$; admissible-kernel restriction & Locates top-down causation in constraint, selection or rule-governed admissibility, not in an added force term \\
$U$ & Event-time maps & Micro tick count, actor-event epoch, learning-update count, macro logging epoch, and any calendar-time index & Prevents cross-scale timing from being confused with causal influence \\
Intervention harness & Baseline, macro intervention and micro translation ensembles & $E_0$, $E_M$, $E_\mu$, matched seeds, matched initial-condition distributions, documented translation rule & Supplies the paired counterfactual design \\
Diagnostics & Outcome and information measures & Price impact, spread, depth, volatility, order-flow imbalance, abstraction error, effective information, transfer entropy and robustness checks & Tests whether macro actor constraints are reducible to micro translations \\
\bottomrule
\end{tabularx}
\end{table*}

\subsection{Environment}
\label{app:environment-and-agents}

The environment is a discrete event-time limit order book. At each micro tick an order, cancellation, execution, or learning update may occur. Trades execute when orders cross, and the price is updated on execution. The simulator must record at least the book state, executed trades, cancellations, agent inventories, cash or wealth variables where relevant, and the event type. The model should not rely on a single global clock. It should record micro ticks and macro actor epochs separately.

The lower-level agent classes should be explicit enough to define the lower-level kernel. Table~\ref{tab:micro-agent-classes} gives the minimum useful agent specification. The exact behavioural rule can be simple, but the state variables and parameters must be logged. Otherwise the intervention does not have a well-defined target.

\begin{table*}[ht!]
\centering
\footnotesize
\caption{Suggested lower-level agent classes and required state variables. These are agents in the lower-level dynamics, not actors in the top-down sense.}
\label{tab:micro-agent-classes}
\begin{tabularx}{\textwidth}{p{0.16\textwidth}p{0.25\textwidth}p{0.24\textwidth}X}
\toprule
Agent class & State variables to log & Parameters or controls & Behavioural role \\
\midrule
Liquidity provider / market maker & Inventory, outstanding bid and ask orders, cash or marked-to-market wealth, recent fills & Spread target, inventory aversion, quote size distribution, cancellation rate, quote refresh rule & Supplies depth near the touch and absorbs order flow subject to inventory risk \\
Fundamental trader & Inventory, cash, perceived fundamental value, signal history, active orders & Signal strength, noise level, risk aversion, order-size rule, execution aggressiveness & Trades when price deviates from a noisy latent or perceived value \\
Noise trader & Inventory, cash, order direction state if persistent, active orders & Market-order arrival rate, buy/sell probability, order-size distribution & Supplies exogenous or weakly structured order flow \\
Adaptive strategist & Inventory, cash, recent profit and loss, observation state, action state, policy or value estimates, exploration state & Learning rate, discount factor, exploration parameter, reward definition, action set, policy-update clock & Learns an execution or aggressiveness policy from market feedback \\
Background liquidity process & Aggregate latent depth, cancellation intensity, exogenous shock state & Intensity parameters, regime state, shock distribution & Provides regime variation without making every source of liquidity an explicit strategic agent \\
\bottomrule
\end{tabularx}
\end{table*}

This table also clarifies the single-agent and many-agent learning cases used later in the appendix. A single adaptive strategist gives the simple learning-agent case. A population of adaptive strategists gives the many-learning-agent case. In both cases the learning state is part of the lower-level state. It cannot be treated as a hidden implementation detail.

\subsection{Actors and constraint interfaces}
\label{app:actors-and-constraint-interfaces}

Actors are first-class macro role bearers. They may issue orders or change rules, but they are not merely large lower-level agents. An actor matters hierarchically because it instantiates a causation class and restricts admissible lower-level kernels through an interface. Table~\ref{tab:actor-classes} gives the minimum actor specification for the market domain.

\begin{table*}[ht!]
\centering
\footnotesize
\caption{Suggested actor classes and their constraint interfaces. These are higher-level actor instances in $C$, not ordinary lower-level agents in $D$.}
\label{tab:actor-classes}
\begin{tabularx}{\textwidth}{p{0.15\textwidth}p{0.17\textwidth}p{0.21\textwidth}p{0.18\textwidth}X}
\toprule
Actor & Candidate causation class & Interface $\eta$ & Target subsystem $B_\ell$ & Example macro intervention and micro translation \\
\midrule
Index rebalancer & Fixed-goal control or algorithmic structuring & Scheduled execution rule, target weights, participation rule, trade list & Order-flow process, execution agents, affected instruments & Change pro-rata execution to time-weighted execution; translate by injecting an aggregate-equivalent order sequence across lower-level agents \\
Regulator / circuit breaker & Algorithmic structuring or fixed-goal control & Price-move threshold, lookback window, pause duration, reopening rule & Matching engine, order acceptance process, cancellation process & Change threshold or pause duration; translate by throttling order arrivals or cancellations without invoking the regulator rule \\
Institutional-flow actor & Fixed-goal control with possible adaptive feedback & Parent-order schedule, participation cap, urgency rule, information leakage rule & Order submission process, liquidity-taking flow, liquidity-provision response & Change persistence or urgency of execution; translate by changing order-arrival intensities or signed-volume profiles \\
Learning-governance actor & Adaptive feedback or adaptive-goal selection & Reward shaping, risk limit, capital constraint, policy-freeze rule, admissibility test & Adaptive strategists and their policy-update kernels & Change reward or risk constraint; translate by changing lower-level learner parameters while holding the macro rule absent \\
Liquidity-regime actor & Adaptive selection or contextual constraint & Regime classification rule, volatility state, liquidity state, activation threshold & Background liquidity process and market-maker quoting kernels & Change regime threshold or activation window; translate by changing static liquidity parameters across agents \\
\bottomrule
\end{tabularx}
\end{table*}

The important requirement is that each actor intervention must identify: the actor state, the target subsystem, the interface, the constrained kernel family, and the proposed micro translation. If any one of these is missing, the experiment cannot distinguish actor-level constraint from ordinary lower-level parameter perturbation.

\subsection{Aggregation}
\label{app:aggregation-clocks-observables}

The aggregation operator must be explicit and logged at the macro epoch. A useful minimal specification is
\[
A_t=G(M_t)=\{P_t,L_t,V_t,F_t\},
\]
where $P_t$ is a price variable such as the midprice after matching, $L_t$ is a liquidity vector such as spread and depth within a fixed number of ticks, $V_t$ is realised volatility over a declared window, and $F_t$ is signed volume or order-flow imbalance per unit event time. Table~\ref{tab:aggregation-and-logging} gives the required logging layer.

\begin{table*}[ht!]
\centering
\footnotesize
\caption{Aggregation and event-time logging requirements. The aggregation map should be fixed before interventions are compared.}
\label{tab:aggregation-and-logging}
\begin{tabularx}{\textwidth}{p{0.17\textwidth}p{0.24\textwidth}p{0.24\textwidth}X}
\toprule
Object & Definition in the simulation & Logged at & Reason required \\
\midrule
Price $P_t$ & Midprice, transaction price, or declared price statistic after matching & Micro ticks and macro epochs & Primary price-impact and volatility input \\
Liquidity $L_t$ & Bid--ask spread, depth within $q$ ticks, queue imbalance, cancellation rate & Micro ticks and macro epochs & Tests whether actor constraints change market resilience and local admissible transitions \\
Volatility $V_t$ & Realised volatility over a declared rolling event window & Macro epochs and reporting windows & Captures regime dependence and circuit-breaker triggers \\
Flow $F_t$ & Net signed volume, order-flow imbalance, or participation-rate measure & Micro ticks and macro epochs & Links actor schedules to lower-level order-flow consequences \\
Micro clock $n_0$ & Count of order arrivals, cancellations, executions or micro events & Every event & Defines lower-level transition order \\
Macro clock $n_1$ & Count of actor evaluations, rebalancing events or regulatory checks & Actor epochs & Defines when constraints are evaluated \\
Learning clock $n_L$ & Count of policy updates or reward-update events & Learning updates & Required when adaptive strategists are active \\
\bottomrule
\end{tabularx}
\end{table*}

This logging layer is not administrative. It is part of the causal design. The same aggregate flow profile can have different effects if it is delivered with different timing, different liquidity state, or different learning-state exposure. The clocks are therefore part of the object being tested, not just implementation detail.

\subsection{Single and many learning-agents}
\label{app:learning-agent-cases}

The learning-agent part should be separated into two cases.

First, a single learning-agent case can be treated as a standard reinforcement-learning problem: one adaptive trader interacts with a market environment, observes a state or signal, chooses actions, receives rewards, and updates a policy or value function \citep{sutton2018reinforcement}. In the market setting this is not only a generic reinforcement-learning abstraction. It has already been used as a concrete stepping stone in which a simple learning agent interacts with an agent-based market model \citep{dicks2024simple}. This case is useful as a baseline because the environment is approximately stationary if the other agents are fixed or non-learning. It also supplies a clean intermediate case between a non-learning agent-based market and the many-learning-agent formulation \citep{dicks2024}.

Second, the many-learning-agent case is different. If many agents learn at the same time, the effective environment of each learner changes as other learners adapt. This is the usual multi-agent reinforcement-learning difficulty: stability of the learning dynamics and adaptation to other changing agents become part of the model rather than nuisance details \citep{busoniu2008comprehensive}. This is the setting closest to many-learning-agent market simulations \citep{dicks2024}.

For hierarchical causality the distinction matters. In the single-agent case an actor intervention may change the agent's experienced environment. In the many-agent case an actor intervention may change the coupled learning process itself. The requirements are therefore stricter when learning is active. The simulation must log policy states, reward definitions, update times, exploration parameters, and whether learning is frozen or active during the intervention window. If these are not logged, the lower-level kernel is underspecified.

\subsection{Actor interventions}
\label{app:actor-interventions-micro-translations}

Each top-down experiment should be posed as a paired intervention. For a given actor instance $\alpha$, define a macro intervention $I_M$ by changing an actor rule or interface. Examples include changing a rebalancer execution schedule, changing a circuit-breaker threshold, or changing the persistence of an institutional execution programme.

For every macro intervention there must be a corresponding micro translation $I_\mu$. The micro translation is not assumed to be correct. It is a conservative attempt to reproduce the same aggregate input by changing lower-level parameters directly. For example, a rebalancer flow can be translated into a sequence of market orders distributed across ordinary agents; a persistent institutional flow can be translated into modified order-arrival intensities; and a circuit-breaker effect can be translated into lower-level throttling rules.

The paired protocol is then:
\begin{enumerate}[label=(\roman*)]
    \item run a baseline ensemble $E_0$;
    \item run a macro actor-intervention ensemble $E_M$;
    \item construct the micro translation $I_\mu$ using only information allowed by the experimental design;
    \item run the micro-translation ensemble $E_\mu$;
    \item compare $E_M$ and $E_\mu$ under matched seeds and matched initial-condition distributions wherever possible.
\end{enumerate}

The design should be biased against easy discovery of top-down autonomy. The micro translation should be as strong as possible without simply reintroducing the actor rule under another name. Otherwise the experiment risks finding autonomy only because the translation was weak.

\subsection{Primary intervention effect}
\label{app:primary-intervention-effect}

The primary test is interventionist. For each observable $Y$, such as price impact, volatility, spread, depth, realised liquidity, drawdown, or order-flow imbalance, compute
\[
\Delta_M(Y)=\mathbb{E}[Y\mid E_M]-\mathbb{E}[Y\mid E_0]
\]
and
\[
\Delta_\mu(Y)=\mathbb{E}[Y\mid E_\mu]-\mathbb{E}[Y\mid E_0].
\]
The relevant object is not the existence of a macro effect alone. The relevant object is the difference
\[
\Delta_M(Y)-\Delta_\mu(Y).
\]
If this difference is stable across seeds, regimes, and micro realisations, then the macro actor intervention is not being captured by the proposed micro translation. This is the direct simulation analogue of the distinction between aggregation equivalence and causal equivalence.

Uncertainty should be reported with bootstrap confidence intervals or comparable resampling intervals, together with distributional two-sample diagnostics. The experiment should report effect sizes, not only significance labels.

\subsection{Abstraction error}
\label{app:abstraction-error}

The second diagnostic should test whether the cross-level description approximately commutes. Causal abstraction work asks when causal models at different resolutions preserve relevant intervention structure under an abstraction map \citep{rubenstein2017causal,beckers2019abstracting}. Here the same idea can be used operationally, without adopting the whole structural-causal formalism.

Let $\mathcal{L}(Y_M)$ be the distribution of a macro observable under the macro actor intervention and let $\mathcal{L}(Y_\mu)$ be the distribution under the micro translation after aggregation through $G$. Define an abstraction error
\[
\varepsilon_{\mathrm{abs}}(Y)
=
d\left(\mathcal{L}(Y_M),\mathcal{L}(Y_\mu)\right),
\]
where $d$ may be a total variation distance, Wasserstein distance, energy distance, kernel two-sample statistic, or another pre-declared distributional discrepancy.

The interpretation is simple. If the macro intervention and the micro translation have matched aggregate inputs but produce different outcome laws, then aggregation-first and constraint-first descriptions are not equivalent for the causal question. This is the non-commuting part of the hierarchy made empirical.

\subsection{Causal-emergence diagnostics}
\label{app:causal-emergence-diagnostics}

Hoel-style causal emergence provides a useful secondary diagnostic.\footnote{Here \(EI\) denotes effective information in the causal-emergence sense of \citeauthor{hoel2013quantifying}: a scale-dependent diagnostic of how
strongly interventions at a chosen scale constrain possible past and future
states. It is used here only as a secondary diagnostic, not as the definition
of hierarchical causality itself \citep{hoel2013quantifying}.} Effective information asks whether interventions at a given scale reduce uncertainty about future states, and whether a macro description can carry more causal information than a micro description \citep{hoel2013quantifying}. Recent surveys of causal emergence also emphasise effective information and related quantities as central measures for connecting emergence and causality in complex systems \citep{yuan2024emergence}.

In this simulation, effective information should not be treated as the definition of hierarchical causality. The primary object remains the actor-induced constraint on admissible lower-level kernels. Effective information is instead a diagnostic of whether the chosen macro description is causally informative. A useful report is therefore
\[
EI_M - EI_\mu
\qquad\mbox{and}\qquad
EI_{\mathrm{macro}}-EI_{\mathrm{micro}},
\]
computed under clearly specified state partitions and intervention distributions. These quantities should be reported only with the discretisation, coarse-graining, and intervention ensemble used to estimate them. Otherwise the values are not comparable.

A strong result would combine three observations: the macro actor intervention differs from the micro translation, the abstraction error is non-negligible, and the macro description has higher effective information for the relevant outcome. The last point strengthens the interpretation, but it does not replace the paired intervention test.

\subsection{Directed information-flow}
\label{app:directed-information-flow}

Directed information-flow measures, such as transfer entropy, can be useful for timing and dependence diagnostics. They can test whether actor signals, rule changes, or macro event indicators carry predictive information for later lower-level variables, and whether lower-level market variables carry predictive information for later actor states.

These measures should be treated cautiously. They measure directed statistical dependence under a chosen history, binning, and conditioning scheme. They do not by themselves establish the constraint-based causal claim. In the requirements specification their role is therefore diagnostic using Transfer Entropy (TE)\footnote{The phrase ``directed information-flow'' is used descriptively here. Transfer entropy follows \citeauthor{schreiber2000transfer} time-series measure based on conditional
transition probabilities, while directed information in the stricter information-theoretic sense is usually associated with \citeauthor{massey1990directed} formulation for causal/feedback channels. Both are diagnostics of directional dependence or information flow; neither replaces the paired intervention test used here
\citep{schreiber2000transfer,massey1990directed}.}:
\[
TE(A\rightarrow M),
\qquad
TE(M\rightarrow A),
\]
where $A$ denotes actor signals or actor states and $M$ denotes selected micro or macro market variables. These should be computed across the baseline, macro intervention, and micro translation ensembles using the same estimator and conditioning set.

\subsection{Robustness}
\label{app:robustness-multiple-realisability}

Multiple realisability is central. The experiment should not depend on one privileged micro implementation. For each macro actor intervention, construct a family of micro realisations
\[
\rho\in\mathcal{R}
\]
that match the relevant aggregate profile but differ in agent composition, liquidity provision, learning-agent mix, or order-size distribution. Then report the stability of effects across this family, for example
\[
\mathrm{Var}_{\rho\in\mathcal{R}}[\Delta_M(Y;\rho)]
\]
and the corresponding distribution of abstraction errors.

The claim is strongest when the macro actor effect is stable across many lower-level realisations, while the detailed micro paths differ. That is the empirical analogue of causal equivalence being stronger than aggregation equivalence.

\subsection{Experimental design}
\label{app:experimental-design-requirements}

A simulation study built from this appendix should pre-specify:
\begin{enumerate}[label=(\roman*)]
    \item the baseline ensemble and all random seed rules;
    \item the actor interventions and their target subsystems;
    \item the micro translations and the information used to construct them;
    \item the aggregation operator $G$ and all macro observables;
    \item the event clocks and synchronisation rules;
    \item the learning-agent update rules, if learning is active;
    \item the primary outcome variables;
    \item the distributional distances and uncertainty procedures;
    \item the robustness regimes and micro-realisation families;
    \item all stopping rules and exclusion rules.
\end{enumerate}

The required output is not a single simulated path. It is an ensemble comparison with enough logging to reconstruct the lower-level kernels, actor states, aggregation maps, and event-time alignment.

\subsection{Interpretation rules}
\label{app:interpretation-rules}

There are three broad outcomes.

If
\[
\Delta_M(Y)\approx \Delta_\mu(Y),
\qquad
\varepsilon_{\mathrm{abs}}(Y)\approx 0,
\]
and no macro-scale information advantage is found, then the proposed actor intervention is practically reducible to the micro translation for that design.

If
\[
\Delta_M(Y)\neq \Delta_\mu(Y),
\qquad
\varepsilon_{\mathrm{abs}}(Y)>0,
\]
and the result is robust across micro realisations, then the actor intervention has distinct constraint effects not captured by aggregation-equivalent micro perturbations. This is the result most directly aligned with hierarchical causality as defined in this paper.

If the result depends strongly on liquidity regimes, learning activity, population composition, or event-time alignment, then the top-down effect is regime-dependent. This is not a failure. It means that the actor constraint has a domain of validity, and the task is to describe that domain rather than erase it.

\subsection{Limitations}
\label{app:simulation-limitations}

The main limitation is that the micro translation is design-dependent. A weak translation can make the macro actor look autonomous too easily. A translation that smuggles the actor rule back into the micro level can make the comparison trivial. This is why the translation rule must be documented and justified.

The second limitation is that information measures depend on partitions, estimators, and intervention distributions. They should be reported as diagnostics, not as standalone proof.

The third limitation is external validity. A limit order book simulator can test the internal coherence of the hierarchical-causality claim, but it does not by itself prove that the same actor constraints operate in empirical markets. The simulation is therefore a bridge from the primer to an experimental programme, not the programme itself.

\section*{Acknowledgements}

Thank you to my friends and colleagues for wonderful conversations. A sincere thank you to George Ellis for introducing me to his thinking around top-down causation some years ago, and many conversations since. Thanks also to Diane Wilcox for conversations and arguments about representation theory, finance, topology and how not to mathematise models. 

\bibliographystyle{elsarticle-harv}
\bibliography{HCP-Refs-v2.4.0}
\end{document}